\documentclass[a4paper]{article}

\usepackage{amsmath}
\usepackage{amssymb}
\usepackage{picins,graphicx}
\usepackage[OT4]{fontenc}  
\usepackage[latin2]{inputenc}
\usepackage{url}
\usepackage[square,numbers]{natbib}
\usepackage{setspace}

\newcommand{\abs}[1]{\left| #1 \right|}
\newcommand{\aabs}[1]{\left\| #1 \right\|}
\newcommand{\okra}[1]{\left( #1 \right)}
\newcommand{\kwad}[1]{\left[ #1 \right]}
\newcommand{\klam}[1]{\left\{ #1 \right\}}

\newcommand{\floor}[1]{\left\lfloor #1 \right\rfloor}
\newcommand{\dzi}[1]{\left\langle #1 \right\rangle}
\newcommand{\sred}[1]{\mathbf{E}\, #1}
\newcommand{\voc}[1]{\mathbf{V}[#1]}

\DeclareMathOperator{\card}{card}

\newcommand{\myeta}{\eta}

\newcommand{\subseq}{\sqsubseteq}

\newtheorem{definition}{Definition}

\newtheorem{theorem}{Theorem}
\newtheorem{lemma}{Lemma}
\newenvironment*{proof}{\begin{trivlist}\item[]
\noindent\textbf{Proof:}}{$\Box$\par\end{trivlist}}
\newenvironment*{proof*}[1]{\begin{trivlist}\item[]
\noindent\textbf{Proof of Theorem #1:}}{$\Box$\par\end{trivlist}}

\renewcommand{\thesection}{\Roman{section}}

\begin{document}

\title{On the Vocabulary of Grammar-Based Codes \\ and the Logical Consistency
  of Texts}

\author{{\L}ukasz D\k{e}bowski%
  \thanks{
    The research reported in this work was supported in part by 
    the Australian Research Council, DP0210999, and the IST Programme
    of the European Community, under the PASCAL II Network of
    Excellence, IST-2002-506778. The material presented in
    Section~\ref{secCodes} of this paper was published in conference
    communication \cite{Debowski07c}. %
    \newline\null\hspace{\parindent}
    {\L}. D\k{e}bowski is with
    the Institute of Computer Science, Polish Academy of Sciences, ul.
    Ordona 21, 01-237 Warszawa, Poland (e-mail: ldebowsk@ipipan.waw.pl).}}

\date{}




\begin{titlepage}
\thispagestyle{empty}   
\titlepage
\maketitle

\begin{abstract}
  The article presents a~new interpretation for Zipf-Mandelbrot's law
  in natural language which rests on two areas of information theory.
  Firstly, we construct a~new class of grammar-based codes and,
  secondly, we investigate properties of strongly nonergodic
  stationary processes.  The motivation for the joint discussion is to
  prove a~proposition with a~simple informal statement: If a~text of
  length $n$ describes $n^\beta$ independent facts in a repetitive way
  then the text contains at least $n^\beta/\log n$ different words,
  under suitable conditions on $n$.  In the formal statement,
  two modeling postulates are adopted.  Firstly, the words are
  understood as nonterminal symbols of the shortest grammar-based
  encoding of the text.  Secondly, the text is assumed to be emitted
  by a~finite-energy strongly nonergodic source whereas the facts are
  binary IID variables predictable in a~shift-invariant way.
  \\[1em]
\end{abstract}


\end{titlepage}
\pagestyle{plain}   


\section{Introduction}
\label{secIntroduction}

\newcommand{\tskip}[1]{[...]}

\begin{quotation}
``If a~Martian scientist sitting before his radio in Mars accidentally
received from Earth the broadcast of an extensive speech \tskip{which he
recorded perfectly through the perfection of Martian apparatus and
studied at his leisure}, what criteria would he have to determine
whether the reception represented the effect of animate process on
Earth, or merely the latest thunderstorm on Earth? It seems that the
only criteria would be the arrangement of occurrences of the elements,
and the only clue to the animate origin would be this: the arrangement
of the occurrences would be neither of rigidly fixed regularity such
as frequently found in wave emissions of purely physical origin nor
yet a~completely random scattering of the same.''
\begin{flushright}
G. K. Zipf \cite[page 187]{Zipf65}
\end{flushright}
\end{quotation}
\medskip

The aim of this paper is to present a~new explanation for the
distribution of words in natural language. To achieve this goal, we
will consider a~new class of grammar-based codes
\cite{KiefferYang00,CharikarOthers05} and we will research
information-theoretic properties of strongly nonergodic stationary
processes.  Thus both linguists and information theorists may find
this paper interesting.

The distribution of words is quite well described by the celebrated
Zipf-Mandelbrot law \cite{Zipf65,Mandelbrot54}, which states that the
word frequency in a~text is an inverse power of the word rank.  Some
effort in applied probability theory has been devoted to inferring
this law for various idealized settings. The most famous explanation
is given by the `monkey-typing' model. In this explanation,
consecutive characters of the text are modeled as independent
identically distributed (IID) variables assuming values of both
letters and spaces whereas the Zipf-Mandelbrot law is obeyed by
strings of letters delimited by spaces
\cite{Mandelbrot54,Miller57}. 
Other published explanations involve, e.g., multiplicative processes
\cite{Simon55} 
and games \cite{HarremoesTopsoe01}.

In this paper, we will consider the integrated version of the
Zipf-Mandelbrot law, usually called Herdan's or Heaps' law in the
English literature. This law says that the number of distinct words
observed in a~text is proportional to a~power of the text length
\cite{KuraszkiewiczLukaszewicz51en,Guiraud54,Herdan64,Heaps78}. The
claim can be inferred from the Zipf-Mandelbrot law assuming certain
regularity of text growth \cite{Khmaladze88,Kornai02}.

The explanation of Herdan's law proposed here is based on previous
partial insights \cite{Debowski06,Debowski07d,
  Debowski09, Debowski10} and addresses two modeling challenges:
\begin{enumerate}
\item Words, in the linguistic sense, can be delimited in the text even
  when the spaces are absent \cite{Wolff80,DeMarcken96,KitWilks99}.
\item Texts, in the linguistic sense, refer to many facts unknown
  a~priori to the reader but they usually do this in a~consistent and
  repetitive way.
\end{enumerate}
Our interest will be focused on proving a~proposition that can be
expressed in the following informal way, assuming thereafter $\beta\in
(0,1)$: \newsavebox{\QLThesis} \savebox{\QLThesis}{(H)}
\begin{itemize}
\item[\usebox{\QLThesis}] If a~text of length $n$ describes $n^\beta$
  independent facts in a repetitive way then the text contains at
  least $n^\beta/\log n$ different words, under appropriate
  quantification over $n$.
\end{itemize}
As we will argue later in this section, some connection can also be
drawn between proposition \usebox{\QLThesis} and the initial
quotation from G.\ K.\ Zipf.

So as to translate proposition \usebox{\QLThesis} into a~provable
statement, we will adopt a~certain mathematical model of words, texts,
and facts that can be motivated linguistically. The main modeling
assumptions are described below. We assume that symbol $\mathbb{N}$
denotes the set of (strictly) positive integers. For a~fixed countable
set $\mathbb{X}$, called the alphabet, we denote the set of nonempty
strings as $\mathbb{X}^+:=\bigcup_{n\in\mathbb{N}} \mathbb{X}^n$ and
the set of all strings as
$\mathbb{X}^*:=\mathbb{X}^+\cup\klam{\lambda}$, where $\lambda$ is the
empty string. The length of a~string $w\in\mathbb{X}^*$, or
$w\in\mathbb{Y}^*$ for any other alphabet $\mathbb{Y}$, will be
written as $\abs{w}$.

\medskip 

\textbf{The number of words in a~text:} Linguists have observed that
strings of characters that are repeated within the text sufficiently
many times often correspond to whole words or set phrases like
\emph{New York}. A~particularly good correspondence is obtained when
word boundaries are detected using a~grammar-based code that minimizes
the length of a~certain text encoding
\cite{Wolff80,DeMarcken96,KitWilks99,NevillManning96}. For that
reason, the number of words in the formalization of proposition
\usebox{\QLThesis} will be modeled with the number of nonterminal
symbols in such an encoding. Let us present some details of this
construction.

Grammar-based codes compress strings by transforming them first into
special grammars, called admissible grammars \cite{KiefferYang00}, and
then encoding the grammars back into strings according to a~fixed
simple method.  An \emph{admissible} grammar is a~context-free grammar
that generates a~singleton language $\klam{w}$ for some string
$w\in\mathbb{X}^*$ \cite{KiefferYang00}.  The subset of such grammars
will be denoted as $\mathcal{G}(w)$, whereas the set of admissible
grammars for all strings is written as
$\mathcal{G}:=\bigcup_{w\in\mathbb{X}^*} \mathcal{G}(w)$.  If the
string $w$ contains repeated substrings then some grammar in
$\mathcal{G}(w)$ `factors out' the repetitions and represents $w$
concisely.

In an admissible grammar, there is exactly one rule per nonterminal
symbol and the nonterminals can be ordered so that the symbols are
rewritten onto strings of strictly succeeding symbols
\cite{KiefferYang00,CharikarOthers05}.  Hence, such a~grammar is given
by its set of production rules
\begin{align}
\label{FullGrammar}
\mathsf{G}&=   \klam{ 
\begin{array}{l}
A_1\rightarrow\alpha_1, \\
A_2\rightarrow\alpha_2, \\
..., \\
A_n\rightarrow\alpha_n 
    \end{array}
}, 
\end{align}
where $A_1$ is the start symbol, other $A_i$ are secondary
nonterminals, and the right-hand sides of rules satisfy $\alpha_i\in
(\klam{A_{i+1},A_{i+2},...,A_n}\cup\mathbb{X})^*$.  

A~concrete example of an admissible grammar is
\begin{align}
\label{exGi}
  \klam{
    \begin{array}{l}
A_{1} \rightarrow A_{2}A_{2}A_{4}A_{5}\textbf{dear\_children}A_{5}A_{3}\textbf{all.} \\
A_{2} \rightarrow A_{3}\textbf{you}A_{5} \\
A_{3} \rightarrow A_{4}\textbf{\_to\_} \\
A_{4} \rightarrow \textbf{Good\_morning} \\
A_{5} \rightarrow \textbf{,\_}
    \end{array}
  }.
\end{align}
If we start the derivation with the symbol $A_1$ and follow the
rewriting rules, we obtain the text of a~song:
\begin{verse}
  \it
  Good morning to you, \\
  Good morning to you, \\
  Good morning, dear children, \\
  Good morning to all.
\end{verse}
In the compressions of longer texts, nonterminal symbols $A_i$ often
correspond to words or set phrases, especially if it is also required
that the nonterminals were defined as strings of only terminal symbols
\cite{KitWilks99}. The latter kind of grammars will be called
\emph{flat grammars}.

The number of distinct nonterminal symbols in grammar
(\ref{FullGrammar}) will be called the \emph{vocabulary size} of
$\mathsf{G}$ and denoted by
\begin{align}
  \voc{\mathsf{G}}:=\card \klam{A_{1},A_{2},...,A_n} =n
  .
\end{align}
On the other hand, a~function
$\Gamma:\mathbb{X}^+\rightarrow\mathcal{G}$ such that
$\Gamma(w)\in\mathcal{G}(w)$ for all $w\in\mathbb{X}^+$ is called
a~\emph{grammar transform} \cite{KiefferYang00}.  In Subsection
\ref{ssecUpperExcess}, we will construct \emph{admissibly minimal}
grammar transforms, which minimize the length of a~specific text
encoding. The vocabulary size of these transforms will be considered
in the formalization of proposition \usebox{\QLThesis}.  The
definition of admissibly minimal transforms is too technical to
present right here but we may say that it resembles transforms
considered by linguists \cite{DeMarcken96,KitWilks99} and the
transform investigated in \cite{CharikarOthers05}, which we call
Yang-Kieffer minimal. In particular, there exist admissibly minimal
transforms that are flat grammar transforms.

\medskip 

In the second turn, we have to formulate a~model of an infinitely long
text that describes random facts in a~repetitive way.  Both the text
and the set of facts repeatedly described in the text will be modeled
as stochastic processes. We introduce quite a~new mathematical model
of human language so we devote more space for motivation.

\medskip 

\textbf{The model of texts and facts:} Let $(X_i)_{i\in\mathbb{Z}}$ be
a~stochastic process on a~probability space $(\Omega,\mathfrak{J},P)$,
where variables $X_i:\Omega\rightarrow\mathbb{X}$ assume values from
the countable alphabet $\mathbb{X}$. This process will model an
infinitely long text, where $X_i$ are consecutive text units. We can
imagine that the values of $X_i$ are characters if $\mathbb{X}$ is
finite, or words or sentences if $\mathbb{X}$ is infinite. In most
cases in this paper, we will assume that $\mathbb{X}$ is finite.  On
the other hand, let $Z_k:\Omega\rightarrow\klam{0,1}$,
$k\in\mathbb{N}$, be equidistributed IID binary variables. We will
assume that the values of $Z_k$ are a~priori unknown to the reader of
the text $(X_i)_{i\in\mathbb{Z}}$ but can be learned from the text.
Variables $Z_k$ will be called facts. We can imagine that the values
of $Z_k$ are logical values (1=true and 0=false) of certain
systematically enumerated logically independent propositions.

More specifically, we will suppose that each fact $Z_k$ can be
inferred from a~half-infinite text if we start reading it from an
arbitrary position. This is done to incorporate the postulate that
facts are described in the text in a~repetitive way.  In the
following, notation $X_{m:n}:=(X_i)_{m\le i\le n}$ will be used for
strings of $X_i$, also called blocks. Blocks $X_{m:n}$ model finite
texts.  This definition, introduced in \cite{Debowski09}, captures
what we need:
\begin{definition}  
  \label{defiUDP}
  A~stochastic process $(X_i)_{i\in\mathbb{Z}}$ is called
  \emph{strongly nonergodic}\footnote{A not so fortunate name
    \emph{uncountable description process} was used originally in
    \cite{Debowski09}.} if there exists an IID binary process
  $(Z_k)_{k\in\mathbb{N}}$ with $P(Z_k=0)=P(Z_k=1)=\frac{1}{2}$ and
  functions $s_k:\mathbb{X}^*\rightarrow \klam{0,1}$,
  $k\in\mathbb{N}$, such that
  \begin{align}
    \label{UDPcondi}
    \lim_{n\rightarrow\infty} P(s_k(X_{t+1:t+n})=Z_k)&=1
    ,
    &
    &\forall t\in\mathbb{Z},\, \forall k\in\mathbb{N} 
    .
  \end{align}
\end{definition}

The motivation for functions $s_k$ comes from the idea that there is
a~fixed method of interpreting finite texts in natural language to
infer facts, which is known as human language competence in linguistic
jargon.  Thus, any fact that is mentioned in texts in a~repetitive way
can be learned by text readers ultimately, regardless of their
starting point.  The facts that are mentioned repeatedly fall roughly
into two types: (i) facts about the unchangeable objective world,
which can be discovered and reported independently by successive
generations of text creators, and (ii) facts about historical
heritage, which undergo distributed creation, accumulation, and
(partly lossy) transmission from text creators to readers.

Definition \ref{defiUDP} is a~mathematical model that ignores
distinction between these two flavors of facts, except for the
requirement that facts cannot change after their discovery or be
forgotten after their creation. Investigating a~relaxed condition is
planned for a~future publication. An enumeration of independent facts
can also be modeled by the binary expansion of halting probability.
The binary expansion of halting probability is algorithmically random
and represents a~large body of mathematical knowledge in its most
condensed form \cite{Chaitin75,Gardner79}. We suppose, however, that
information relayed by humans in a~repetitive way is mostly unrelated
to this theoretical concept.

From the probabilistic point of view, a~stationary process is strongly
nonergodic if and only if there exists a~continuous random variable
$Y:\Omega\rightarrow(0,1)$ measurable with respect to the
shift-invariant $\sigma$-algebra \cite[Theorem 9]{Debowski09}. Such
a~variable is an example of a~parameter in Bayesian statistics. For
instance, taking $Y=\sum_{k\in\mathbb{N}} 2^{-k} Z_k$ corresponds to
a~uniform prior on $Y$. Theorem 9 from \cite{Debowski09} has a~few
consequences.  Firstly, a~strongly nonergodic process cannot be
ergodic, or IID in particular.  Secondly, it cannot be a~finite-state
hidden Markov process, which is a~kind of processes considered in the
`monkey-typing' explanations of Zipf-Mandelbrot's law.

However, to illustrate how the concept of a~strongly nonergodic
process matches some preconceptions about human communication, let us
consider the following example. It is simple but very different from
parametric models usually considered in statistics.  For a~while, let
the alphabet be $\mathbb{X}=\mathbb{N}\times\klam{0,1}$ and let the
process $(X_i)_{i\in\mathbb{Z}}$ have the form
\begin{align}
  \label{exUDPi}
  X_i:=(K_i,Z_{K_i}),
\end{align}
where $(Z_k)_{k\in\mathbb{N}}$ and $(K_i)_{i\in\mathbb{Z}}$ are
probabilistically independent whereas $(K_i)_{i\in\mathbb{Z}}$ is such
an ergodic stationary process that $P(K_i=k)>0$ for every natural
number $k\in\mathbb{N}$.  Under such assumptions it can be
demonstrated that $(X_i)_{i\in\mathbb{Z}}$ forms a~strongly nonergodic
process.\footnote{In spite of a~few years of acquaintance, I~have not
  found a~plausible scientific name for process (\ref{exUDPi}).
  Probably it should be called simply the Santa Fe process because I
  discovered it during a~visit to the Santa Fe Institute.}

Variables $X_i=(K_i,Z_{K_i})$ can be given such a~linguistic
interpretation: Imagine that $(X_i)_{i\in\mathbb{Z}}$ is a~sequence of
consecutive statements extracted from an infinitely long text that
describes an infinite random object $(Z_k)_{k\in\mathbb{N}}$
consistently.  Each statement $X_i=(k,z)$ reveals both the address $k$
of a~random bit of $(Z_k)_{k\in\mathbb{N}}$ and its value $Z_k=z$.
Logical consistency of the description is reflected in this property:
If two statements $X_i=(k,z)$ and $X_j=(k',z')$ describe bits of the
same address ($k=k'$) then they always assert the same bit value
($z=z'$). Let us note that the pool of facts $(Z_k)_{k\in\mathbb{N}}$
can be viewed either as a~random state of an objective world that
exists prior to the text $(X_i)_{i\in\mathbb{Z}}$ or as historical
heritage that is created on-line and memorized during generation of
consecutive variables $X_i$. Model (\ref{exUDPi}) is indifferent with
respect to either interpretation.

In the formalization of proposition \usebox{\QLThesis}, the number of
facts described in the finite text $X_{1:n}$ will be identified with
the number of $Z_i$'s that may be predicted with probability at least
$\delta$ given $X_{1:n}$. That is, this number will be understood as
the cardinality of set
\begin{align}
\label{Un}
  U_\delta(n):= \klam{k\in\mathbb{N}:
    P\okra{s_k\okra{X_{1:n}}=Z_k}\ge \delta}
  ,
\end{align}
where $\delta>\frac{1}{2}$.  As we will show in Subsection
\ref{ssecExample}, the cardinality of set $U_\delta(n)$ is of order
$n^\beta$ for process (\ref{exUDPi}) if variables $K_i$ are IID and
power-law distributed,
\begin{align}
  \label{ZetaK}
  P(K_i=k)&=k^{-1/\beta}/\zeta(\beta^{-1})
  ,
  &
  \beta&\in(0,1)
  ,
\end{align}
where $\zeta(x)=\sum_{k=1}^\infty k^{-x}$ is the zeta function.  In
contrast, it can be seen that the cardinality of $U_\delta(n)$ is of
order $\log n$ if $(X_i)_{i\in\mathbb{Z}}$ is a~Bernoulli process with
a~random parameter $Y=\sum_{k\in\mathbb{N}} 2^{-k} Z_k$. Note that the
cardinality of $U_\delta(n)$ for a~given process depends, to a~certain
extent, on the choice of functions $s_k$ and facts $Z_k$.  The
formalization of proposition \usebox{\QLThesis} holds, however,
for any choice of $s_k$ and $Z_k$ as long as (\ref{UDPcondi}) is
satisfied.

\medskip

Now we can approach the main result.  Let $\sred{}$ be the expectation
operator and let $\card{A}$ be the cardinality of a~set $A$.  We also
use this concept from \cite{Shields97}:
\begin{definition}
  A~process $(X_i)_{i\in\mathbb{Z}}$ is called
  a~\emph{finite-energy process} if
  \begin{align}
    \label{FE}
    P(X_{t+\abs{w}+1:t+\abs{wu}}=u|X_{t+1:t+\abs{w}}=w)\le Kc^{\abs{u}}
  \end{align}
  for all $t\in\mathbb{Z}$, all $u,w\in\mathbb{X}^*$, and certain
  constants $c<1$ and $K$, as long as $P(X_{t+1:t+\abs{w}}=w)>0$.
\end{definition}
It can be easily seen that stationary finite-energy processes have
a~positive entropy rate. Moreover, condition (\ref{FE}) is satisfied
for processes dithered with an IID noise \cite{Shields97}---so it
seems reasonable in modeling natural language.

Our formalization of proposition \usebox{\QLThesis} takes the
following form:
\begin{theorem}
  \label{theoQLThesis}
  Let $(X_i)_{i\in\mathbb{Z}}$ be a stationary finite-energy strongly
  nonergodic process over a~finite alphabet $\mathbb{X}$. Assume that
  inequality
  \begin{align}
    \label{QLPremise}
    \liminf_{n\rightarrow\infty} \frac{\card{U_\delta(n)}}{n^\beta}>0
  \end{align}
  holds for some $\beta\in (0,1)$, $\delta\in(\frac{1}{2},1)$, and
  sets (\ref{Un}) where functions $s_{k}$ satisfy (\ref{UDPcondi}).
  Then
  \begin{align}
    \label{QLClaim}
    \limsup_{n\rightarrow\infty} 
    \sred{\okra{\frac{
          \voc{\Gamma(X_{1:n})}
        }{
          n^\beta(\log n)^{-1}
        }}^p}>0
    ,
    \quad
    p>1
    ,
  \end{align}
  for any admissibly minimal grammar transform
  $\Gamma:\mathbb{X}^+\rightarrow\mathcal{G}$.
\end{theorem}
As we will see in Subsection \ref{ssecExample}, an example of
a~process over a~finite alphabet that satisfies the premise of Theorem
\ref{theoQLThesis} can be constructed by stationary coding of the
process (\ref{exUDPi}) with $K_i$ satisfying (\ref{ZetaK}), cf.\
\cite{Debowski10}.  

Theorem \ref{theoQLThesis} is closely related to two propositions
pertaining to mutual information between two adjacent blocks. For
a~discrete stationary process $(X_i)_{i\in\mathbb{Z}}$, let us define
the \emph{$n$-symbol entropy}
\begin{align}
  \label{Hn}
  H(n):=H(X_{t+1:t+n})
  =-\sred{\log P(X_{t+1:t+n})}
  ,
\end{align}
where $\log$ is the natural logarithm.  Denote the block mutual
information as
\begin{align}
  \label{En}
  E(n)&:=2H(n)-H(2n)
  =
  I(X_{1:n};X_{n+1:2n}),
\end{align}
called the \emph{$n$-symbol excess entropy} in
\cite{CrutchfieldFeldman03}. $E(n)$ is a~convenient measure of
long-range dependence in discrete-valued processes. We have:
\begin{theorem}
  \label{theoQLThesisA}
  Let $(X_i)_{i\in\mathbb{Z}}$ be a stationary strongly nonergodic
  process over a~finite alphabet $\mathbb{X}$.  Assume that inequality
  (\ref{QLPremise}) holds for some $\beta\in (0,1)$,
  $\delta\in(\frac{1}{2},1)$, and sets (\ref{Un}) where functions
  $s_{k}$ satisfy (\ref{UDPcondi}). Then
\begin{align}
  \label{QLClaimA}
  \limsup_{n\rightarrow\infty} \frac{E(n)}{n^\beta}>0
  .
\end{align}
\end{theorem}
\begin{theorem}
  \label{theoQLThesisB}
  Let $(X_i)_{i\in\mathbb{Z}}$ be a stationary finite-energy process
  over a~finite alphabet $\mathbb{X}$.  Assume that inequality
  \begin{align}
    \label{QLPremiseB}
    \liminf_{n\rightarrow\infty} \frac{E(n)}{n^\beta}>0
  \end{align}
  holds for some $\beta\in(0,1)$. Then we have (\ref{QLClaim}) for any
  admissibly minimal grammar transform
  $\Gamma:\mathbb{X}^+\rightarrow\mathcal{G}$.
\end{theorem}

Although Theorem \ref{theoQLThesis} does not follow from Theorems
\ref{theoQLThesisA} and \ref{theoQLThesisB}, we will give almost
a~simultaneous proof of all three propositions. A~heuristic proof of
Theorem \ref{theoQLThesisB} was outlined in \cite{Debowski06}.  This
paper provides the formal proof and develops a~discussion of the
logically earlier Theorem \ref{theoQLThesisA}. Because of space
constraints, we do not discuss hypothetical extensions of Theorem
\ref{theoQLThesis} such as strong laws.

The proper discussion of the linguistic relevance of our results is
also beyond the scope of this paper and will be presented in later
publications. However, let us note that the conjecture $E(n)\propto
n^\beta$ was raised for natural language by Hilberg \cite{Hilberg90}.
This was his interpretation of the graph of conditional entropy in
Shannon's seminal paper \cite{Shannon50} and he supposed that
$\beta\approx \frac{1}{2}$. This conjecture is little known among
linguists but has evoked a~discussion about `statistical complexity'
among physicists \cite{EbelingNicolis91,
  EbelingPoschel94,
  BialekNemenmanTishby01b, CrutchfieldFeldman03}.  In our opinion,
Theorem \ref{theoQLThesisA} demonstrates that Hilberg's hypothesis can
be motivated rationally, whereas Theorem \ref{theoQLThesisB} shows
that the hypothesis implies certain empirical regularities.  The
initial quotation from G.\ K.\ Zipf matches these results
qualitatively since processes with $E(n)\propto n^\beta$ differ from
both regular oscillations and memoryless noise. Indeed, our
preliminary experiments indicate that the vocabulary size of
admissibly minimal codes is much larger for texts in natural language
than for memoryless sources \cite{Debowski07d}.

The further composition of this paper is as follows:
Section~\ref{secMain} contains the proof of Theorems
\ref{theoQLThesis}, \ref{theoQLThesisA}, and \ref{theoQLThesisB}.
In Section~\ref{secCodes}, we define admissibly minimal grammar
transforms and build a~new class of universal grammar-based codes
associated with those transforms.  Section~\ref{secUDP} is a~study of
nonergodic stationary processes. It contains two results used to prove
Theorem \ref{theoQLThesisA} and an exposition of a~process that
satisfies the assumption of Theorem \ref{theoQLThesis}.  The article
is supplemented with two appendices.  In Appendix
\ref{ssecExcessBound}, we bound the expression $2G(n)-G(2n)$ for
a~nonnegative function $G$ that has a~vanishing linear rate of
growth. Appendix \ref{ssecMaximalRepeat} provides an upper bound for
the expected length of a~repeat in a~block sampled from
a~finite-energy process.

\section{The proof of Theorems \ref{theoQLThesis}--\ref{theoQLThesisB}}
\label{secMain}

The proof rests on several intermediate results developed later in
this paper.  Let $H(n)$ be the $n$-symbol  entropy of the
stationary process $(X_i)_{i\in\mathbb{Z}}$, defined in (\ref{Hn}).
The entropy rate
\begin{align}
\label{h}
  h:=\inf_{n\in\mathbb{N}} H(n)/n=\lim_{n\rightarrow\infty} H(n)/n
\end{align}
is another important parameter of the process \cite{CoverThomas91}.

Consecutively, we will use a~convenient shorthand
\begin{align}
  \label{HUn}
  H^U(n):=hn+\kwad{\log 2 -\myeta(\delta)} \cdot \card{U_\delta(n)}
  ,
\end{align}
where $\myeta(p)$ denotes the entropy of distribution $(p,1-p)$,
\begin{align}
  \label{BinaryEntropy}
  \myeta(p) &:= -p\log p-(1-p)\log (1-p)
  .
\end{align}
By Theorem \ref{theoUDPVoc} from Subsection \ref{ssecFacts}, we have
\begin{align}
  \label{HUineq}
  H(n)\ge H^U(n)
\end{align}
and
\begin{align}
  \label{HUrates}
  \lim_{n\rightarrow\infty} H(n)/n=h=
  \lim_{n\rightarrow\infty} H^U(n)/n
\end{align}
for a~strongly nonergodic process over a~finite alphabet.

Consider next the $n$-symbol excess entropy $E(n)=2H(n)-H(2n)$. From
(\ref{HUineq}) and (\ref{HUrates}) we obtain
\begin{align}
  \label{DiffEUBeta}
  \liminf_{n\rightarrow\infty} \frac{\card{U_\delta(n)}}{n^\beta}>0
  \implies 
  \limsup_{n\rightarrow\infty} \frac{E(n)}{n^\beta}>0
\end{align}
as an instance of implication (\ref{DiffFGBeta}) from Appendix
\ref{ssecExcessBound}. This proves Theorem \ref{theoQLThesisA}.

Now let us proceed to prove the claims that involve grammar
transforms. For an admissibly minimal grammar transform $\Gamma$, let
$C=B(\Gamma(\cdot)):\mathbb{X}^+\rightarrow \mathbb{Y}^+$,
$\mathbb{Y}=\klam{0,1,...,D_Y-1}$, be the associated grammar-based
code, defined in Subsection \ref{ssecUpperExcess} (Definition
\ref{defiAdmissibly}). Denote the expected length of the code $C$ as
\begin{align}
  \label{HCn}
  H^C(n)&:=\sred{\abs{C(X_{1:n})}} \log D_Y .
\end{align}
This code is uniquely decodable (i.e., its extension
$C^*:(u_1,...,u_k)\mapsto C(u_1)...C(u_k)$ is an injection), so we
have the source coding inequality
 \begin{align}
  \label{HCineq}
  H^C(u)\ge H(n)
  .
\end{align}
Moreover, by Theorem \ref{theoUni} from Subsection
\ref{ssecUniversal}, code $C$ is nearly universal, i.e.,
\begin{align}
  \label{HCrates}
  \lim_{n\rightarrow\infty} H^C(n)/n=h=
  \lim_{n\rightarrow\infty} H(n)/n
\end{align}
for any stationary finite-energy process over a~finite alphabet.

Consider the expected excess length $E^C(n):=2H^C(n)-H^C(2n)$ of the
code $C$. Relations (\ref{HCineq}) and (\ref{HCrates}) yield
\begin{align}
  \label{DiffECEBeta}
  \liminf_{n\rightarrow\infty} \frac{E(n)}{n^\beta}>0
  \implies 
  \limsup_{n\rightarrow\infty} \frac{E^C(n)}{n^\beta}>0
\end{align}
as an instance of implication (\ref{DiffFFBeta}) from Appendix
\ref{ssecExcessBound}.

Moreover, for a~stationary finite-energy strongly nonergodic process
over a~finite alphabet there holds a~double inequality
\begin{align}
  \label{HIIIineq}
  H^C(u)\ge H(n)\ge H^U(n)
\end{align}
and an equality of rates
\begin{align}
  \label{HIIIrates}
  \lim_{n\rightarrow\infty} H^C(n)/n=
  \lim_{n\rightarrow\infty} H(n)/n=
  \lim_{n\rightarrow\infty} H^U(n)/n
  .
\end{align}
Using implication (\ref{DiffFGBeta}) again, relations (\ref{HIIIineq})
and (\ref{HIIIrates}) yield respectively
\begin{align}
 \label{DiffECUBeta}
  \liminf_{n\rightarrow\infty} \frac{\card{U_\delta(n)}}{n^\beta}>0
  \implies 
  \limsup_{n\rightarrow\infty} \frac{E^C(n)}{n^\beta}>0
  .
\end{align}

To upper-bound the excess length of the code in terms of the
vocabulary size, denote the maximal length of a~(possibly overlapping)
repeat in $w$ as
\begin{align}
\label{DefL}
  \mathbf{L}(w)&:=\max 
  \klam{|s|: w=x_1sy_1=x_2sy_2 \land x_1\ne x_2}
  ,
\end{align}
where $s,x_i,y_i\in\mathbb{X}^*$.  Then, by Theorem
\ref{theoUpper}(i)--(ii) from Subsection \ref{ssecUpperExcess}, we
have
\begin{align}
\label{UpperVocII}
\hspace{-0.5em}
\abs{C(u)}+\abs{C(v)}-\abs{C(w)}
&\le  W_0\voc{\Gamma(w)}(1+\mathbf{L}(w))
\end{align}
for $w=uv$ and a~certain constant $W_0$.  In the following, define
\begin{align*}
  S_n&:=\voc{\Gamma(X_{1:2n})}n^{-\beta}\log n
  ,
  \\
  T_n&:=(1+\mathbf{L}(X_{1:2n})) (\log n)^{-1}
  .
\end{align*} 
Inequality (\ref{UpperVocII}) and H\"older's inequality yield
\begin{align*}
  E^C(n)\, n^{-\beta} 
  \le W_0\, \sred{S_n T_n}
  \le W_0 (\sred{S_n^p})^{1/p} (\sred{T_n^q})^{1/q}
\end{align*}
for $p,q>1$ such that $(p-1)(q-1)=1$. 

Since $\sred{T_n^q}$ are bounded above for a~finite-energy process by
Lemma \ref{theoFEL} from Appendix \ref{ssecMaximalRepeat},
consecutively we have
\begin{align}
  \label{DiffVLECBeta}
  \limsup_{n\rightarrow\infty} \frac{E^C(n)}{n^\beta}>0
  \implies 
  \limsup_{n\rightarrow\infty} \sred{S_n^p}>0
  .
\end{align}
Theorem \ref{theoQLThesis} follows from propositions
(\ref{DiffECUBeta}) and (\ref{DiffVLECBeta}), whereas Theorem
\ref{theoQLThesisB} is implied by propositions (\ref{DiffECEBeta}) and
(\ref{DiffVLECBeta}).


\section{Grammar-based codes}
\label{secCodes}

For the set of admissible grammars $\mathcal{G}$,
a~\emph{grammar-based code} is a~uniquely decodable code of form
$C=B(\Gamma(\cdot)):\mathbb{X}^+\rightarrow \mathbb{Y}^+$, where
$\Gamma:\mathbb{X}^+\rightarrow\mathcal{G}$ is a~(string-to-)grammar
transform and $B:\mathcal{G}\rightarrow\mathbb{Y}^+$ is called
a~\emph{grammar(-to-string) encoder} \cite{KiefferYang00}. To
guarantee existence of universal codes of this form, we will assume
in this section that both the input and output alphabets are finite,
$\mathbb{X}=\klam{0,1,...,{D_X}-1}$ and
$\mathbb{Y}=\klam{0,1,...,{D_Y}-1}$ in particular.

We are interested in finding a~class of nearly universal grammar-based
codes for which the excess code length
\begin{align}
  \label{ExcessCode}
  \abs{C(u)}+\abs{C(v)} -\abs{C(uv)}
\end{align}
can be bounded by the vocabulary size $\voc{\Gamma(w)}$ and the
maximal length (\ref{DefL}) of a~repeat in $w=uv$.  Let us note that
a~similar bound can be obtained for the excess grammar length
\begin{align}
  \label{ExcessGrammar}
  \abs{\Gamma(u)}+\abs{\Gamma(v)} -\abs{\Gamma(uv)}
\end{align}
of certain grammar transforms. In this expression, the
\emph{Yang-Kieffer length} of an admissible grammar is defined as
\begin{align}
\label{YKlength}
\abs{\mathsf{G}}:=\textstyle\sum_i \abs{\alpha_i}
\end{align}
for grammar (\ref{FullGrammar}) 
\cite{KiefferYang00}. We have:
\begin{theorem}
\label{theoMinYK}
Let $\Gamma$ be a~\emph{Yang-Kieffer minimal grammar transform}, i.e.,
\begin{align}
\label{MinYKG}
  \abs{\Gamma(w)}&=\min_{\mathsf{G}\in\mathcal{G}(w)} |\mathsf{G}|
\end{align}
and $\alpha_i\not=\lambda$ for any secondary rule
$A_i\rightarrow\alpha_i$ in $\Gamma(w)$.
For any strings $u,v\in\mathbb{X}^*$ and $w=uv$ we have
\begin{align}
\label{MinYK}
  0\le
  \abs{\Gamma(u)}+\abs{\Gamma(v)}-\abs{\Gamma(w)}\le
  \voc{\Gamma(w)}\mathbf{L}(w)
  .
\end{align}
\end{theorem}

This result was noticed in part in \cite[Theorem 3]{Debowski06}. To
motivate further constructions, let us present a~proof of the right
inequality. 
\begin{proof}  
  Let the grammar $\mathsf{G}=\Gamma(w)$ for $w=uv$ be of form
  (\ref{FullGrammar}).  We will split it into two grammars for $u$ and
  $v$, respectively
\begin{align*}
\mathsf{G}_L&=   \klam{ 
\begin{array}{l}
A_1\rightarrow x_Ly_L, \\
A_2\rightarrow\alpha_2, \\
..., \\
A_n\rightarrow\alpha_n 
    \end{array}
}
, 
&
\mathsf{G}_R&=   \klam{ 
\begin{array}{l}
A_1\rightarrow y_Rx_R, \\
A_2\rightarrow\alpha_2, \\
..., \\
A_n\rightarrow\alpha_n 
    \end{array}
} 
,
\end{align*}
where $y_L,y_R\in\mathbb{X}^*$ and either $\alpha_1=x_Lx_R$ or
$\alpha_1=x_LA_ix_R$ for some secondary nonterminal $A_i$.  By
minimality of $\Gamma$, each secondary nonterminal $A_i$ must appear
at least twice on the right-hand sides of rules in $\mathsf{G}$.
(Otherwise, we could find a~strictly shorter grammar than $\mathsf{G}$
by deleting $A_i$ from the grammar.)  Hence we have $|y_Ly_R| \le
\mathbf{L}(w)$ and $|\alpha_i| \le \mathbf{L}(w)$ for $i\ge 2$. Thus
we obtain
\begin{align*}
  \abs{\Gamma(u)}+\abs{\Gamma(v)}\le |\mathsf{G}_L|+|\mathsf{G}_R| 
  \le \abs{\mathsf{G}} + n\cdot\mathbf{L}(w).
\end{align*} 
Regrouping the terms yields the right inequality in (\ref{MinYK}).
\end{proof}

There exists a~grammar encoder
$B_\text{YK}:\mathcal{G}\rightarrow\mathbb{Y}^+$
such that $C=B_\text{YK}(\Gamma(\cdot))$ is a~universal code for any
Yang-Kieffer minimal grammar transform $\Gamma$ \cite{KiefferYang00}.
Unfortunately, for this encoder, it is hard to relate the excess
grammar length (\ref{ExcessGrammar}) to the excess code length
(\ref{ExcessCode}). Thus we will consider another encoder.

Subsequently,  notation (\ref{FullGrammar}) will be reduced to
\begin{align}
\label{GrammarNotation}
  \mathsf{G}=(\alpha_1,\alpha_2,...,\alpha_n)
  .
\end{align}
We will define a~grammar encoder that represents grammar $\mathsf{G}$
as a~string resembling list (\ref{GrammarNotation}). This encoder
yields nearly universal codes for grammar transforms that minimize the
length of the code over a~sufficiently large subclass of grammars
(Theorem \ref{theoUni}). Since the encoder provides nearly
a~homomorphism between some operations on grammars and strings, the
respective codes satisfy also Theorem \ref{theoUpper}---an analogue of
Theorem \ref{theoMinYK}. Such codes and grammar transforms are called
\emph{admissibly minimal} and are defined in Subsection
\ref{ssecUpperExcess}.

\subsection{Local encoders and minimal transforms}
\label{ssecEncodersTranforms}

The proof of Theorem \ref{theoMinYK} applies certain `cut-and-paste'
operations on grammars. For example, an operation resembling the
following joining operation was used in \cite{Debowski06} to prove the
left inequality in (\ref{MinYK}):
\begin{definition}
  For any function $f:\mathbb{U}\rightarrow\mathbb{W}^*$,
  where concatenation on domains $\mathbb{U}^*$ and $\mathbb{W}^*$ is
  defined, denote its \emph{extension to strings}
  $f:\mathbb{U}^*\rightarrow\mathbb{W}^*$ as
  \begin{align}
    \label{Extension}
    f^*(x_1x_2...x_m):= f(x_1)f(x_2)...f(x_m)
    ,
  \end{align}
  where $x_i\in\mathbb{U}$.  Next, for grammars
  $\mathsf{G}_i=(\alpha_{i1},\alpha_{i2},...,\alpha_{in_i})$, $i=1,2$,
  define the \emph{joining} of $\mathsf{G}_1$ and $\mathsf{G}_2$ as
  \begin{align*}
    \mathsf{G}_1\oplus \mathsf{G}_2:=(A_2A_{n_1+2},\,
    &H_1^*(\alpha_{11}),H_1^*(\alpha_{12}),...,H_1^*(\alpha_{1n_1}),
    \\
    &H_2^*(\alpha_{21}),H_2^*(\alpha_{22}),...,H_2^*(\alpha_{2n_2})),
  \end{align*}
  where $H_1(A_j)=A_{j+1}$ and $H_2(A_j)=A_{j+n_1+1}$ for nonterminals
  and $H_1(x)=H_2(x)=x$ for terminals $x\in\mathbb{X}$.  
\end{definition}
We have $\mathsf{G}_1\oplus \mathsf{G}_2\in\mathcal{G}(uv)$ if
$\mathsf{G}_1\in\mathcal{G}(u)$ and $\mathsf{G}_2\in\mathcal{G}(v)$.

Now we need such a~grammar encoder
$B:\mathcal{G}\rightarrow\mathbb{Y}^+$ that the edit distance between
$B(\mathsf{G}_1\oplus \mathsf{G}_2)$ and
$B(\mathsf{G}_1)B(\mathsf{G}_2)$ is small. In the following
construction, the set of positive integers $\mathbb{N}$ is treated as
a~generic infinite countable alphabet with concatenation $ab$,
addition $a+b$, and subtraction $a-b$.
\begin{definition}
  For the set of terminals $\mathbb{X}=\klam{0,1,...,{D_X}-1}$,
$B:\mathcal{G}\rightarrow\mathbb{Y}^+$ is
called a~\emph{local grammar encoder} if
\begin{align}
\label{LocalCoder}
B(\mathsf{G})=B_\text{S}^*(B_\text{N}(\mathsf{G})), 
\end{align}
where:
\begin{enumerate}
\item the function
  $B_\text{N}:\mathcal{G}\rightarrow(\klam{0}\cup\mathbb{N})^*$
  encodes grammars as strings of integers so that the encoding of
  a~grammar $\mathsf{G}=(\alpha_1,\alpha_2,...,\alpha_n)$ is the
  string
  \begin{align*}
   B_\text{N}&(\mathsf{G}):= 
   \\
   &F_1^*(\alpha_1){D_X}F_2^*(\alpha_2){D_X}...{D_X}F_n^*(\alpha_n)({D_X}+1), 
  \end{align*}
which employs identity transformation $F_i(x)=x$ for terminals
$x\in\mathbb{X}$ and relative indexing $F_i(A_j)={D_X}+1+j-i$ for
nonterminals,
\item the function $B_\text{S}:\klam{0}\cup\mathbb{N}
  \rightarrow\mathbb{Y}^+$ is an injection, the set
  $B_\text{S}(\klam{0}\cup\mathbb{N})$ is prefix-free, and the length
  function $|B_\text{S}(\cdot)|$ is nondecreasing.---We call such
  a~$B_\text{S}$ a~\emph{natural number encoder}.
\end{enumerate}
\end{definition}
The local encoder $B$ is uniquely decodable and the edit distance
between $B(\mathsf{G}_1\oplus \mathsf{G}_2)$ and
$B(\mathsf{G}_1)B(\mathsf{G}_2)$ is small indeed.  For instance, if
$B(\mathsf{G}_i)=\gamma_iB_\text{S}({D_X}+1)$ then
$B(\mathsf{G}_1\oplus \mathsf{G}_2)=\delta \gamma_1 B_\text{S}({D_X})
\gamma_2 B_\text{S}({D_X}+1)$ where $\delta=B_\text{S}({D_X}+2)
B_\text{S}({D_X}+2+\voc{\mathsf{G}_1}) B_\text{S}({D_X})$.

Subsequently, let us introduce grammar transforms that minimize the
length of code $\abs{B(\cdot)}$ over subclasses of admissible
grammars.  A~subclass $\mathcal{J}\subset\mathcal{G}$ will be called
\emph{sufficient} if there exists a~grammar transform
$\Gamma:\mathbb{X}^+\rightarrow\mathcal{J}$, i.e., if $
\mathcal{G}(w)\cap\mathcal{J}\not=\emptyset$ for all
$w\in\mathbb{X}^+$.
\begin{definition}
  For an arbitrary function
  $\aabs{\cdot}:\mathcal{G}\rightarrow\klam{0}\cup\mathbb{N}$ (called
  later length) and a~sufficient subclass $\mathcal{J}$, a~grammar
  transform $\Gamma:\mathbb{X}^+\rightarrow\mathcal{J}$ will be called
  a~\emph{$(\aabs{\cdot},\mathcal{J})$-minimal grammar transform} if
  $\aabs{\Gamma(w)}\le \aabs{\mathsf{G}}$ for all
  $\mathsf{G}\in\mathcal{G}(w)\cap \mathcal{J}$ and
  $w\in\mathbb{X}^+$. If $\aabs{\cdot}=\abs{B(\cdot)}$, the respective
  code $B(\Gamma(\cdot))$ will be called
  \emph{$(B,\mathcal{J})$-minimal}.
\end{definition}

The following sufficient subclasses of grammars will be discussed
later.
\begin{definition}
  We say that $(\alpha_1,\alpha_2,...,\alpha_n)$ is a~\emph{flat
    grammar} if $\alpha_i\in \mathbb{X}^+$ for $i\ge 2$.  The set of
  flat grammars is denoted as $\mathcal{F}$.  Secondly,
  $\mathcal{D}_k\subset\mathcal{F}$ denotes the class of
  \emph{$k$-block interleaved grammars}, i.e., grammars
  $(\alpha_1,\alpha_2,...,\alpha_n)\in\mathcal{F}$ where $\alpha_i\in
  \mathbb{X}^k$ for $i\ge 2$.  Thirdly,
  $\mathcal{B}_k\subset\mathcal{D}_k$ stands for the set of
  \emph{$k$-block grammars}, i.e., grammars
  $(uw,\alpha_2,...,\alpha_n)\in\mathcal{D}_k$ where each
  $A_{2},A_{3},...,A_n$ appears in the string
  $u\in(\klam{A_{2},A_{3},...,A_n})^*$ whereas the string
  $w\in\mathbb{X}^*$ has length $|w|<k$, cf.\ \cite{NeuhoffShields98}.
  Finally, we put the class of block grammars
  $\mathcal{B}:=\bigcup_{k\ge 1} \mathcal{B}_k$ and the class of block
  interleaved grammars $\mathcal{D}:=\bigcup_{k\ge 1} \mathcal{D}_k$.
\end{definition}
Flat grammar transforms $\Gamma:\mathbb{X}^+\rightarrow\mathcal{F}$
were used to detect word boundaries in the computational linguistic
experiment by \cite{KitWilks99}.

\subsection{Universal codes for local encoders}
\label{ssecUniversal}

Local encoders resemble the encoder considered by Neuhoff and Shields
in \cite{NeuhoffShields98}, denoted here as $B_\text{NS}$.  The
authors have established that any $(B_\text{NS},\mathcal{B})$-minimal
code is universal for the class of block grammars $\mathcal{B}$, and
we will use this fact to prove that certain codes employing local
encoders are nearly universal.  The main difference between the
encoder $B_\text{NS}$ and a~local encoder is that $B_\text{NS}$
encodes a~nonterminal $A_i$ as a~string of length $\log_{D_Y}
\voc{\mathsf{G}}$ whereas the local encoder uses a~string of length
$|B_\text{S}({D_X}+i)|$.  This is not a~big difference so we can prove
the following proposition using some results of
\cite{NeuhoffShields98}.
\begin{theorem}
\label{theoUni}
For the set of terminals $\mathbb{X}=\klam{0,1,...,{D_X}-1}$ and the
output alphabet $\mathbb{Y}=\klam{0,1,...,{D_Y}-1}$, let
$B:\mathcal{G}\rightarrow\mathbb{Y}^+$ be a~local grammar encoder
(\ref{LocalCoder}) that satisfies
\begin{align}
\label{UniCoder}
\limsup_{n\rightarrow\infty} \abs{B_\text{S}(n)}/\log_{D_Y} n =1.
\end{align}
Then for any sufficient subclass of grammars
$\mathcal{J}\supset\mathcal{B}$, every $(B,\mathcal{J})$-minimal code
$C$ is nearly strongly universal, i.e.,
\begin{align}
  \label{StronglyUni}
  \limsup_{n\rightarrow\infty}
  \frac{\abs{C(X_{1:n})}\log D_Y}{n}\le h \quad \text{a.s.}
\end{align}
for any stationary ergodic process $(X_k)_{k\in\mathbb{Z}}$ over the
alphabet $\mathbb{X}$ with an entropy rate $h>0$.%
\footnote{This theorem was mistakenly stated in \cite{Debowski07c}
  without the assumption $h>0$.}
\end{theorem}
\emph{Remark 1:} A~natural number encoder $B_\text{S}$ such that
(\ref{UniCoder}) holds can be chosen, e.g., as the Elias $D_Y$-ary
representation $\omega:\klam{0}\cup\mathbb{N}\rightarrow\mathbb{Y}^*$
\cite{Elias75}, where
$$\abs{\omega(n)} =
\begin{cases}     
  1  &\text{if } n < D_Y, \\
  \abs{\omega(\floor{\log_{D_Y} n})} 
  + \floor{\log_{D_Y} n} + 1  &\text{if } n \ge D_Y.  
\end{cases}$$
\\
\emph{Remark 2:} Claim (\ref{StronglyUni}) may be generalized to
finite-energy processes as follows, cf.\ \cite{Weissman05}.  Let
$(X_k)_{k\in\mathbb{Z}}$ be a~stationary finite-energy process over
the alphabet $\mathbb{X}$ and let $h_F$ be the entropy rate of the
process's random ergodic measure, viz.  (\ref{Paramh}) and
(\ref{Decomp}).  Firstly, from the finite-energy property and the
Shannon-McMillan-Breiman theorem it follows that $h_F>0$. Hence, the
strong ergodic decomposition theorem \cite[a statement in the proof of
Theorem 9.12]{Kallenberg97} and the claim of Theorem \ref{theoUni}
imply
\begin{align} 
  \label{StronglyUniF}
  \limsup_{n\rightarrow\infty} 
  \frac{\abs{C(X_{1:n})}\log D_Y}{n}\le h_F 
  \quad \text{a.s.}
\end{align}
Secondly, since $0\le \abs{C(X_{1:n})}\le K n$ for a~$K>0$, inequality
(\ref{StronglyUniF}) implies
\begin{align}
  \label{WeaklyUni}
  \lim_{n\rightarrow\infty} 
  \sred{\okra{\frac{\abs{C(X_{1:n})}\log D_Y}{n}}}=h
\end{align}
by equality $h=\sred{h_F}$ \cite{GrayDavisson74b} and the inverse
Fatou lemma.
\begin{proof}
  Consider a~sequence of $k$-block grammar transforms
  $\Gamma_k:\mathbb{X}^+\rightarrow\mathcal{B}_k$, $k\ge 1$. Because
  $|B_\text{S}(\cdot)|$ is nondecreasing, we have
  \begin{align*}
    \abs{B(\Gamma_k(w))}\le 
    \alpha k\voc{\Gamma_{k}(w)} + 
    \frac{n}{k}
    \abs{B_\text{S}(D_X+\voc{\Gamma_{k}(w)})}
  \end{align*}
  for an $\alpha>0$. On the other hand, for an $\epsilon>0$ and
  a~stationary ergodic process $(X_k)_{k\in\mathbb{Z}}$ with an
  entropy rate $h$, let $k(n)$ be the largest integer $k$ satisfying
  $ke^{k(h+\epsilon)}\le n$. Neuhoff and Shields showed in
  \cite{NeuhoffShields98} that, for $h\ge 0$,
  \begin{align*}
    \limsup_{n\rightarrow\infty} \max_{w\in\mathbb{X}^n} 
    \frac{\log \voc{\Gamma_{k(n)}(w)}}{k(n)} &\le h+2\epsilon
    ,
    \\
    \lim_{n\rightarrow\infty} 
    \voc{\Gamma_{k(n)}(X_{1:n})}\cdot  k(n)/n &= 0
    \quad \text{a.s.}
  \end{align*}
  Notice that $\lim_n k(n)=\infty$. Hence $\lim_n
  \voc{\Gamma_{k(n)}(X_{1:n})}=\infty$ a.s.\ for $h>0$ by
  \cite[Theorem 2 on page 912]{OrnsteinWeiss90}. 
  Moreover, (\ref{UniCoder}) holds, and thus we obtain
  \begin{align*}
    \limsup_{n\rightarrow\infty} 
    \frac{\abs{B(\Gamma_{k(n)}(X_{1:n})}\log D_Y}{n}\le h+2\epsilon
    \quad \text{a.s.}
  \end{align*}
  for $h>0$. Hence any $(B,\mathcal{J})$-minimal code is nearly
  strongly universal.
\end{proof}

\subsection{Bounds for the  vocabulary size}
\label{ssecUpperExcess}

Now we will derive the analogue of Theorem \ref{theoMinYK} for some
minimal grammar-based codes that use local grammar encoders. Firstly,
the code lengths are almost subadditive.  Secondly, the excess code
lengths are dominated by the vocabulary size multiplied by the length
of the longest repeat. To show this, we will introduce a~few other
operations on grammars.
\begin{definition}
  \label{defiCutPaste}
  Consider a~grammar
  \begin{align}
    \label{FullGrammarSh}
    \mathsf{G}=(\alpha_1,\alpha_2,...,\alpha_n)\in\mathcal{G}(w)
    .
  \end{align}
  For $0\le p,q\le \abs{w}$ and $p+q=\abs{w}$, let the strings
  $u,v\in\mathbb{X}^*$ satisfy $p=\abs{u}$, $q=\abs{v}$ and $uv=w$.
  Define then the \emph{left} and \emph{right croppings} of
  $\mathsf{G}$ as
  \begin{align*}
    \mathbb{L}_p \mathsf{G}:=(x_Ly_L,\alpha_2,...,\alpha_n)\in\mathcal{G}(u),
    \\
    \mathbb{R}_q \mathsf{G}:=(y_Rx_R,\alpha_2,...,\alpha_n)\in\mathcal{G}(v),
  \end{align*}
  where $y_L,y_R\in\mathbb{X}^*$ and either $\alpha_1=x_Lx_R$ or
  $\alpha_1=x_LA_ix_R$ for some secondary nonterminal $A_i$. 

  The \emph{expansion} $\dzi{\alpha}_\mathsf{F}$ of a~string
  $\alpha\in (\klam{A_{1},A_{2},A_{3},...,A_n}\cup\mathbb{X})^*$ with
  respect to a~subset of rules $\mathsf{F}\subset \mathsf{G}$ of the
  grammar (\ref{FullGrammarSh}) is the unique element of the language
  generated by grammar $\klam{A_0\rightarrow\alpha}\cup\mathsf{F}$
  with the start symbol $A_0$, cf.\ \cite{CharikarOthers05}.  Define
  then
  \begin{enumerate}
  \item the \emph{nonterminal deleting}
    \begin{align*}
      \mathbb{U}_i\mathsf{G}:=
      (&\phi_i(\alpha_1),
      \phi_i(\alpha_2),
      ...
      \phi_i(\alpha_{i-1}),
      \\
      &\phi_i(\alpha_{i+1}),
      \phi_i(\alpha_{i+2}),
      ...
      \phi_i(\alpha_{n}))
      ,
    \end{align*}
    where
    $\phi_i(\alpha)=H_i^*(\dzi{\alpha}_{\klam{A_i\rightarrow\alpha_i}})$,
    $H_i(A_j)=A_{j-1}$ for $j>i$, $H_i(A_j)=A_{j}$ for $j<i$, $H_i(x)=x$
    for $x\in\mathbb{X}$, and $2\le i\le n$,
  \item the \emph{flattening}
    \begin{align*}
      \mathbb{F}\mathsf{G}:=
      (\alpha_1,\dzi{\alpha_2}_\mathsf{G},\dzi{\alpha_3}_\mathsf{G},...,
      \dzi{\alpha_n}_\mathsf{G})
      ,
    \end{align*}
  \item and the \emph{secondary part}
    \begin{align*}
      \mathbb{S}\mathsf{G}:=
      (\lambda,\alpha_2,\alpha_3,...,\alpha_n)
      .
    \end{align*}
  \end{enumerate}
\end{definition}
\begin{theorem}
\label{theoUpper}
For the set of terminals $\mathbb{X}=\klam{0,1,...,{D_X}-1}$, let
$B:\mathcal{G}\rightarrow\mathbb{Y}^+$ be a~local grammar encoder
(\ref{LocalCoder}). Introduce constants
\begin{align}
\label{UpperLowerConst}
W_n&:=\abs{B_\text{S}({D_X}+1+n)}.
\end{align}
Let $\Gamma$ be a~$(\aabs{\cdot},\mathcal{J})$-minimal grammar
transform for the length $\aabs{\cdot}=\abs{B(\cdot)}$.  Consider the
code $C=B(\Gamma(\cdot))$ and strings $u,v,w\in\mathbb{X}^+$.
\begin{enumerate}
\item If $\mathsf{G}\in\mathcal{J}\implies \mathbb{L}_p \mathsf{G},\,
  \mathbb{R}_q \mathsf{G},\, \mathbb{U}_i \mathsf{G}\in\mathcal{J}$
  for all valid $p,q,i$ then
  \begin{align}
    \label{UpperLeftRight}
    \hspace{-2em} \abs{C(u)},\, \abs{C(v)} &\le \abs{C(uv)} +
    W_0\mathbf{L}(uv),
    \\
    \label{UpperUpper}
    \hspace{-2em}
    \abs{C(u)}+\abs{C(v)}-\abs{C(uv)} 
    &\le \aabs{\mathbb{S}\Gamma(uv)}+W_0\mathbf{L}(uv).
  \end{align}
\item If $\mathsf{G}\in\mathcal{J}\implies \mathbb{F}\mathsf{G},\,
  \mathbb{U}_i \mathsf{G}\in\mathcal{J}$ for all valid $i$ then
  \begin{align}
    \label{UpperVoc}
    \aabs{\mathbb{S}\Gamma(w)}+W_0\mathbf{L}(w)\le
    W_0\voc{\Gamma(w)}(1+\mathbf{L}(w)).
  \end{align}
\item If $\mathsf{G}_1,\mathsf{G}_2\in\mathcal{J}\implies
  \mathsf{G}_1\oplus \mathsf{G}_2\in\mathcal{J}$ then
  \begin{align}
    \abs{C(u)}+&\abs{C(v)}-\abs{C(uv)} 
    \nonumber
    \\
    &\ge -W_0-W_1-W_{\voc{\Gamma(u)}+1}.
    \label{UpperLower}
  \end{align}
\end{enumerate}
\emph{Remark:} In particular, the premises of proposition (i)--(ii)
are satisfied for
$\mathcal{J}= \mathcal{G},\mathcal{F},\mathcal{D},\mathcal{D}_k$
whereas  the premise of proposition (iii) is satisfied for
$\mathcal{J}=\mathcal{G}$ 
Moreover, inequalities (\ref{UpperUpper}) and (\ref{UpperVoc}) imply
together inequality (\ref{UpperVocII}), which is an analogue of
(\ref{MinYK}).%
\footnote{Propositions (i) and (ii) were mistakenly stated in
  \cite{Debowski07c} without the condition $\mathbb{U}_i
  \mathsf{G}\in\mathcal{J}$.}
\end{theorem}
\begin{proof}
  If $\mathsf{G}\in\mathcal{J}\implies \mathbb{U}_i
  \mathsf{G}\in\mathcal{J}$ for all valid $i$ then each secondary
  nonterminal $A_i$ must appear at least twice on the right-hand sides
  of rules in $\Gamma(w)$ whereas the right-hand side of rule
  $(A_i\rightarrow\alpha_i)\in\Gamma(w)$ may not be empty. (Otherwise,
  we would obtain $\aabs{\Gamma(w)}>\aabs{\mathbb{U}_i \Gamma(w)}\in
  \mathcal{J}$ because $\abs{B_\text{S}(n)}>0$ and
  $\abs{B_\text{S}(\cdot)}$ is nondecreasing.) Hence we have
  $\abs{\dzi{\alpha_i}_\mathsf{G}} \le \mathbf{L}(w)$ for $i\ge 2$.
  This result is used to prove propositions (i) and (ii) below.
\begin{enumerate}
\item Set $p=\abs{u}$, $q=\abs{v}$, and $w=uv$.  The claimed
  inequalities follow from
\begin{align*}
\aabs{\Gamma(w)}+W_0\mathbf{L}(w)&\ge
  \aabs{\mathbb{L}_p\Gamma(w)}\ge \aabs{\Gamma(u)}, 
\\
\aabs{\Gamma(w)}+W_0\mathbf{L}(w)&\ge
  \aabs{\mathbb{R}_q\Gamma(w)}\ge \aabs{\Gamma(v)}, 
\end{align*}
and
\begin{align*}
  \aabs{\mathbb{L}_p\Gamma(w)}+&\aabs{\mathbb{R}_q\Gamma(w)}
  \\
  &\le
  \aabs{\Gamma(w)}+ \aabs{\mathbb{S}\Gamma(w)}+W_0\mathbf{L}(w)
  .
\end{align*}
\item  The claim is entailed by $\aabs{\mathbb{S}\Gamma(w)}\le
  \aabs{\mathbb{S}\mathbb{F}\Gamma(w)}$ and
  \begin{align*}
    \aabs{\mathbb{S}\mathbb{F}\Gamma(w)}\le
    W_0\okra{\voc{\Gamma(w)}-1}(1+\mathbf{L}(w))+W_0
    .
  \end{align*}
\item The result is implied by $\aabs{\Gamma(uv)}\le
  \aabs{\Gamma(u)\oplus\Gamma(v)}$ and 
  \begin{align*}
    \aabs{\mathsf{G}_1\oplus\mathsf{G}_2} \le
    W_1+W_{\voc{\mathsf{G}_1}+1}+W_0+\aabs{\mathsf{G}_1}+\aabs{\mathsf{G}_2}
    ,
  \end{align*}
  where $\mathsf{G}_1=\Gamma(u)$ and $\mathsf{G}_2=\Gamma(v)$.
\end{enumerate}
\end{proof}

The strengths of Theorems \ref{theoUni} and \ref{theoUpper}(i)--(ii)
can be combined for the following class of codes and grammar
transforms:
\begin{definition}
  \label{defiAdmissibly}
  A~grammar transform $\Gamma:\mathbb{X}^+\rightarrow\mathcal{G}$,
  where $\mathbb{X}=\klam{0,1,...,{D_X}-1}$, and the associated code
  $C=B(\Gamma(\cdot)):\mathbb{X}^+\rightarrow \mathbb{Y}^+$, where
  $\mathbb{Y}=\klam{0,1,...,{D_Y}-1}$, are called \emph{admissibly
    minimal} if
  \begin{enumerate}
  \item $\Gamma$ is a~$(\abs{B(\cdot)},\mathcal{J})$-minimal grammar
    transform, where
  \item $B:\mathcal{G}\rightarrow\mathbb{Y}^+$ is a~local grammar
    encoder (\ref{LocalCoder}) that satisfies (\ref{UniCoder}),
  \item $\mathcal{J}\supset\mathcal{B}$ for the subclass of block
    grammars $\mathcal{B}$,
  \item $\mathsf{G}\in\mathcal{J}\implies \mathbb{F}\mathsf{G},\,
    \mathbb{L}_p \mathsf{G},\, \mathbb{R}_q \mathsf{G},\, \mathbb{U}_i
    \mathsf{G}\in\mathcal{J}$ for all valid $p,q,i$.
  \end{enumerate}
\end{definition}
\emph{Remark:} In particular, we may take
$\mathcal{J}=\mathcal{G},\mathcal{F},\mathcal{D}$.

\section{Strongly nonergodic processes}
\label{secUDP}

In this section we explore stationary processes rather than codes.
The main goal is to demonstrate equality (\ref{HUrates}) and
inequality (\ref{HUineq}) for strongly nonergodic processes over
a~finite alphabet. The proofs are given in Subsection \ref{ssecFacts}.
This is followed by a~construction of a~process that satisfies the
assumption of Theorem \ref{theoQLThesis}, given in Subsection
\ref{ssecExample}.

\subsection{A bound for the number of facts}
\label{ssecFacts}

For this subsection we need a measure-theoretic generalization of
mutual information, cf.\ \cite{GelfandKolmogorovYaglom56en}.  For
a~probability space $(\Omega,\mathfrak{J},P)$, a~\emph{partition} of
the $\sigma$-algebra $\mathfrak{J}\subset 2^\Omega$
is a~finite set of events $\klam{B_j}_{j=1}^J$ such that
$B_j\in\mathfrak{J}$, $B_i\cap B_j=\emptyset$, and $\bigcup_{j=1}^J
B_j=\Omega$. We define \emph{mutual information} between partitions
$\alpha=\klam{A_i}_{i=1}^I$ and $\beta=\klam{B_j}_{j=1}^J$ with
respect to probability measure $P$ as
\begin{align}
  I_P(\alpha;\beta)
  &:=\sum_{i=1}^I\sum_{j=1}^J
  P(A_i\cap B_j) \log\frac{P(A_i\cap B_j)}{P(A_i)P(B_j)}
  \label{MIp}
  ,
\end{align}
where $0\log 0/x:=0$.

Now, let $\mathfrak{A}$, $\mathfrak{B}$, and $\mathfrak{C}$ be
subalgebras of $\sigma$-algebra $\mathfrak{J}$. That is,
$\klam{\emptyset,\Omega}\subset
\mathfrak{A},\mathfrak{B},\mathfrak{C}\subset \mathfrak{J}$ as well as
$\mathfrak{A}$, $\mathfrak{B}$, and $\mathfrak{C}$ are closed w.r.t.\
operations $\cap$, $\cup$, and $\setminus$. Moreover let the random
variable $P(A||\mathfrak{C})$ be the conditional probability of event
$A\in\mathfrak{J}$ w.r.t.\ the smallest $\sigma$-algebra containing
$\mathfrak{C}$ \cite[Section 33]{Billingsley79}.  We may extend the
concepts of \emph{conditional mutual information}, \emph{mutual
  information}, \emph{conditional entropy}, and \emph{entropy}
respectively as
\begin{align}
  I(\mathfrak{A};\mathfrak{B}|\mathfrak{C})&:=
  \sup_{\alpha\subset\mathfrak{A},\beta\subset\mathfrak{B}}
  \sred{I_{P(\cdot||\mathfrak{C})}(\alpha;\beta)}
  \label{CMIa}
  ,
  \\
  I(\mathfrak{A};\mathfrak{B})&:=
  I(\mathfrak{A};\mathfrak{B}|\klam{\emptyset,\Omega})
  ,
  \\
  H(\mathfrak{A}|\mathfrak{C})&:=
  I(\mathfrak{A};\mathfrak{A}|\mathfrak{C})
  ,
  \\
  H(\mathfrak{A})&:=
  I(\mathfrak{A};\mathfrak{A}|\klam{\emptyset,\Omega})
  \label{Ha}
  ,
\end{align}
where we write $\beta=\klam{B_j}_{j=1}^J\subset\mathfrak{B}$ if and
only if all $B_j\in\mathfrak{B}$, cf.\
\cite{Debowski09,Dobrushin59en,Pinsker60en}.
These concepts generalize the corresponding definitions for random
variables. If we consider discrete random variables $Y_i$ and the
smallest subalgebras $\mathfrak{A}_i\subset\mathfrak{J}$ such that all
events of form $(Y_i=y_i)$ belong to $\mathfrak{A}_i$, then
$I(Y_1;Y_2|Y_3)=I(\mathfrak{A}_1;\mathfrak{A}_2|\mathfrak{A}_3)$,
$I(Y_1;Y_2)=I(\mathfrak{A}_1;\mathfrak{A}_2)$,
$H(Y_1|Y_3)=H(\mathfrak{A}_1|\mathfrak{A}_3)$, and
$H(Y_1)=I(\mathfrak{A}_1)$. Moreover, quantities
(\ref{CMIa})--(\ref{Ha}) satisfy familiar chain rules and enjoy
certain continuity \cite{Dobrushin59en}, \cite{Pinsker60en},
\cite[Theorems 1 and 2]{Debowski09}.

Consider a~stationary process $(X_i)_{i\in\mathbb{Z}}$, where
$X_i:(\Omega,\mathfrak{J})\rightarrow (\mathbb{X},\mathfrak{X})$.
Let $(\mathbb{X}^{\mathbb{Z}},\mathfrak{X}^{\mathbb{Z}})$ be the
measurable space of double infinite sequences.  For the shift
transformation $T: \mathbb{X}^{\mathbb{Z}}\ni
(x_k)_{k\in\mathbb{Z}}\mapsto
(x_{k+1})_{k\in\mathbb{Z}}\in\mathbb{X}^{\mathbb{Z}}$, where
$x_k\in\mathbb{X}$, define the \emph{shift-invariant algebra}
$\mathfrak{I}_\mathbb{X}:=\klam{A\in \mathfrak{X}^{\mathbb{Z}}:
  TA=A}$.
Let $(\mathbb{S},\mathfrak{S})$ be the measurable space of stationary
probability measures on
$(\mathbb{X}^{\mathbb{Z}},\mathfrak{X}^{\mathbb{Z}})$ (i.e., $\mu\circ
T=\mu$ for $\mu\in\mathbb{S}$) and let
$(\mathbb{E},\mathfrak{E})\subset (\mathbb{S},\mathfrak{S})$ be the
subspace of \emph{ergodic} measures (i.e., $\mu(A)\in\klam{0,1}$ for
$\mu\in\mathbb{E}$ and $A\in\mathfrak{I}_\mathbb{X}$).  Precisely,
$\mathfrak{S}$ and $\mathfrak{E}$ are defined as the smallest
$\sigma$-algebras containing all cylinder sets
$\klam{\mu\in\mathbb{S}:\mu(A)\le r}$ and
$\klam{\mu\in\mathbb{E}:\mu(A)\le r}$,
$A\in\mathfrak{X}^{\mathbb{Z}}$, $r\in\mathbb{R}$, respectively.

For an arbitrary measure $\mu\in\mathbb{S}$, let us denote its
$n$-symbol entropy and entropy rate
\begin{align}
  \label{ParamHn}
  H_\mu(n)&:=-\sred{\log \mu(\xi_{t+1:t+n})},
  \\
  \label{Paramh}
  h_\mu&:=\lim_{n\rightarrow\infty} H_\mu(n)/n
  ,
\end{align}
where $\xi_i:\mathbb{X}^{\mathbb{Z}}\ni
(x_k)_{k\in\mathbb{Z}}\mapsto x_i\in\mathbb{X}$ are random variables
on the space of $(\mathbb{X}^{\mathbb{Z}},\mathfrak{X}^{\mathbb{Z}})$.
In particular, $H_\mu(n)=H(n)$ and $h_\mu=h$ when
$\mu=P((X_k)_{k\in\mathbb{Z}}\in\cdot)$.  

Now we will consider $n$-symbol entropy and entropy rate of another
measure associated with the stationary process. Let us put
\begin{align}
  \label{F}
  \mathfrak{F}:=(X_i)_{i\in\mathbb{Z}}^{-1}(\mathfrak{I}_\mathbb{X})\subset
  \mathfrak{J}
  .
\end{align}
According to the ergodic decomposition theorem \cite[Theorems
9.10-12]{Kallenberg97}, if $\mathbb{X}$ is countable, then there
exists a random ergodic measure $F:(\Omega,\mathfrak{F})\rightarrow
(\mathbb{E},\mathfrak{E})$ such that
\begin{align}
  \label{Decomp}
  F(A)=P((X_i)_{i\in\mathbb{Z}}\in A||\mathfrak{F})
\end{align}
for all $A\in \mathfrak{X}^{\mathbb{Z}}$.
Having introduced this measure,
we obtain:
\begin{theorem}
  \label{theoEDecompF}
  For a~stationary process $(X_i)_{i\in\mathbb{Z}}$ over a~finite
  alphabet $\mathbb{X}$, 
\begin{align}
  \label{HDecompF}
  H(n)&= I(X_{1:n};\mathfrak{F})+\sred{H_F(n)}
  ,
\end{align}
where $\sred{H_F(n)}\ge hn$.
\end{theorem}
\begin{proof}
  We have $H(n)=H(X_{1:n})$ and $\sred{H_F(n)}=H(X_{1:n}|\mathfrak{F})$
  by (\ref{Decomp}). Hence (\ref{HDecompF}) follows by the chain rule
  $H(\mathfrak{A})=I(\mathfrak{A};\mathfrak{B})+
  H(\mathfrak{A}|\mathfrak{B})$ \cite[Section 3.6]{Pinsker60en},
  \cite[Theorem 2(ii)]{Debowski09}.  On the other hand, inequality
  $\sred{H_F(n)}\ge hn$ follows from $H_\mu(n)\ge h_\mu n$ and
  equality $h=\sred{h_F}$ derived by \cite{GrayDavisson74b} for
  stationary processes over a~finite alphabet.
\end{proof}

This yields the needed results:
\begin{theorem}
\label{theoUDPVoc}
For a~stationary strongly nonergodic process $(X_i)_{i\in\mathbb{Z}}$
over a~finite alphabet $\mathbb{X}$, inequality (\ref{HUineq}) and
equality (\ref{HUrates}) are satisfied for $\delta\in(\frac{1}{2},1)$.
\end{theorem}
\begin{proof}
  By continuity of mutual information \cite[Section
  2.2]{Pinsker60en}, \cite[Theorems 1(v) and 2(i)]{Debowski09},
\begin{align}
I\okra{X_{1:n};(Z_k)_{k\in\mathbb{N}}}
&= \lim_{k\rightarrow\infty} I(X_{1:n};Z_{1:k})
\nonumber
\\
&= \sum_{k=1}^\infty I(X_{1:n};Z_k|Z_{1:k-1}).  
\label{preUDPVoc}
\end{align}
On the other hand, 
\begin{align*}
  I(X_{1:n};Z_k|Z_{1:k-1}) 
&=H(Z_k|Z_{1:k-1})-H(Z_k|X_{1:n},Z_{1:k-1})
\\
&\ge \log 2 - H(Z_k|s_k(X_{1:n}))
\\
&\ge \log 2 -  \myeta(P(s_k(X_{1:n})=Z_k))
\end{align*}
by the Fano inequality $H(Y_1|Y_2)\le \myeta(P(Y_1=Y_2))$ for a~binary
variable $Y_2$ \cite[Theorem
2.11.1]{CoverThomas91}. 
Restricting the summation in (\ref{preUDPVoc}) to $k\in U_\delta(n)$
yields
\begin{align*}
I\okra{X_{1:n};(Z_k)_{k\in\mathbb{N}}} 
\ge \kwad{\log 2 -\myeta(\delta)}\cdot\card U_\delta(n)
\end{align*}
because $\myeta(\delta)\ge \myeta(P(s_k(X_{1:n})=Z_k))$ for
$\delta\ge\frac{1}{2}$ and $k\in U_\delta(n)$.  Since all events of
form $(Z_k=0)$ and $(Z_k=1)$ belong to the completion of algebra
$\mathfrak{F}$ by \cite[Theorem 9]{Debowski09}, we have
$I(X_{1:n};\mathfrak{F})\ge I\okra{X_{1:n};(Z_k)_{k\in\mathbb{N}}}$ by
the data processing inequality \cite[Theorem 1(iv)]{Debowski09}. In
the following, by Theorem \ref{theoEDecompF} we obtain
\begin{align*}
  H(n)&\ge hn+I(X_{1:n};\mathfrak{F})
  \\
  &\ge
  hn+I\okra{X_{1:n};(Z_k)_{k\in\mathbb{N}}}
  \ge
  H^U(n)
  .
 \end{align*}
 By $H(n)\ge H^U(n)\ge hn$, (\ref{HUrates}) holds as well.
\end{proof}

\subsection{An example of a process}
\label{ssecExample}

In this subsection we will present a~process that satisfies the
assumptions of Theorem \ref{theoQLThesis}. The process will be denoted
as $(\bar Y_i)_{i\in\mathbb{Z}}$ and will be given by stationary
coding of the process (\ref{exUDPi}).  The requirements for the
process $(\bar Y_i)_{i\in\mathbb{Z}}$ are as follows:
\begin{itemize}
\item[(a)] $(\bar Y_i)_{i\in\mathbb{Z}}$ is a~process over a~finite
  alphabet $\mathbb{Y}$,
\item[(b)] $(\bar Y_i)_{i\in\mathbb{Z}}$ is stationary,
\item[(c)] $(\bar Y_i)_{i\in\mathbb{Z}}$ has finite energy, and
\item[(d)] there exists an IID binary process $(\bar
  Z_k)_{k\in\mathbb{N}}$ with $\bar P(\bar Z_k=0)=\bar P(\bar
  Z_k=1)=\frac{1}{2}$, and functions $\bar s_k:\mathbb{Y}^*\rightarrow
  \klam{0,1}$, $k\in\mathbb{N}$, such that
  \begin{align}
    \label{UDPcondY}
    \lim_{n\rightarrow\infty} \bar P\okra{\bar s_{k}\okra{\bar
        Y_{t+1:t+n}}=\bar Z_k}&=1
    ,
    &
    &\forall t\in\mathbb{Z},\, \forall k\in\mathbb{N} 
    ,
  \end{align}
  and
  \begin{align}
    \label{UDPPowerLawY}
    \liminf_{n\rightarrow\infty} \frac{\card{\bar U_{\bar\delta}(n)}}{n^{\beta}}>0
  \end{align}
  for a~certain $\beta\in (0,1)$, all
  ${\bar\delta}\in(\frac{1}{2},1)$, and sets
  \begin{align}
    \bar U_{\bar\delta}(n):= \klam{k\in\mathbb{N}: \bar
      P\okra{\bar s_{k}\okra{\bar Y_{1:n}}=\bar Z_k}\ge {\bar\delta}}
    .
  \end{align}
\end{itemize}

The analogues of properties (b)--(d), but not (a), are satisfied by
the process (\ref{exUDPi}) with $K_i$ being IID and satisfying
(\ref{ZetaK}).  For example, let us derive (\ref{QLPremise}).  We will
write $u\subseq v$ when an infinite sequence or a~string $v$ contains
a~string $u$ as a~substring.  For
$\mathbb{X}=\mathbb{N}\times\klam{0,1}$ and
$v\in\mathbb{X}^{\mathbb{Z}}\cup\mathbb{X}^*$, define the predictors
$s_k$ as
\begin{align}
  \label{funcUDP}
  s_{k}(v):=
  \begin{cases}
    0 & \text{if $(k,0)\subseq v$ and $(k,1)\not\subseq v$}, \\
    1 & \text{if $(k,1)\subseq v$ and $(k,0)\not\subseq v$}, \\
    2 & \text{else}.
  \end{cases}
\end{align}
We have $Z_k=s_k((X_i)_{i\in\mathbb{Z}})$ almost surely.  Then let
$U_\delta(n)$ be the set of well predictable facts, defined in
(\ref{Un}). In view of equality
\begin{align*}
P(s_k(X_{1:n})=Z_k)
&=P(\text{$K_i=k$ for some $i\in\klam{1,...,n}$})
\\
&=1-[1-P(K_i=k)]^n,
\end{align*}
we have $k\in U_\delta(n)$ if and only if $P(K_i=k)\ge
1-(1-\delta)^{1/n}$. This yields
$$
U_\delta(n)\supset \klam{k\in\mathbb{N}: P(K_i=k)\ge -n^{-1}\log
  (1-\delta)} 
$$
by inequality $1-x^{1/n}\le -n^{-1}\log x$ for $x> 0$. Hence
\begin{align}
  \label{UnZeta}
  \card U_\delta(n)\ge
  \kwad{\frac{n}{-\zeta(\beta^{-1})\log(1-\delta)}}^{\beta}
  .
\end{align}
In particular, (\ref{QLPremise}) is satisfied.

Now we have to explain what stationary coding is.  Firstly, consider
a~function $f:\mathbb{X}\rightarrow\mathbb{Y}^*$ that maps single
symbols into strings. We define its
\emph{extension to double infinite sequences}
$f^{\mathbb{Z}}:\mathbb{X}^{\mathbb{Z}}\rightarrow
\mathbb{Y}^{\mathbb{Z}}\cup(\mathbb{Y}^*\times\mathbb{Y}^*)$ as
\begin{align}
  \label{InfExtension}
  f^{\mathbb{Z}}((x_i)_{i\in\mathbb{Z}})&:=
  ... f(x_{-1})f(x_{0})\textbf{.}f(x_1)f(x_2)...
  ,
\end{align}
where $x_i\in\mathbb{X}$.  (The bold-face dot separates the $0$-th and
the first symbol.)  Secondly, for the shift operation
$T((x_i)_{i\in\mathbb{Z}}):=(x_{i+1})_{i\in\mathbb{Z}}=
...x_{0}x_{1}\textbf{.}x_2x_3...$, a~measure $\mu$ on
$(\mathbb{X}^{\mathbb{Z}},\mathfrak{X}^{\mathbb{Z}})$ is called
asymptotically mean stationary (AMS) if limits
\begin{align}  
  \label{BarMu}
  \bar\mu(A)=\lim_{n\rightarrow\infty} \frac{1}{n}  \sum_{i=0}^{n-1}
  \mu\circ T^{-i}(A)
\end{align}
exist for all $A\in\mathfrak{X}^{\mathbb{Z}}$, cf.\
\cite{GrayKieffer80}. The limit $\bar\mu$, if it exists as a~total
function $\mathfrak{X}^{\mathbb{Z}}\rightarrow \mathbb{R}$, forms
a~stationary measure on
$(\mathbb{X}^{\mathbb{Z}},\mathfrak{X}^{\mathbb{Z}})$, i.e.,
$\bar\mu\circ T^{-1}=\bar\mu$, and is called the stationary mean of
$\mu$.  Every stationary measure is AMS \cite{GrayKieffer80}.
Moreover, for an AMS measure $\mu$, the transported measure $\mu\circ
\okra{f^{\mathbb{Z}}}^{-1}$ is AMS under mild conditions, cf.\
\cite[Example 6]{GrayKieffer80}, \cite{Debowski10}. The stationary
mean $\overline{\mu\circ \okra{f^{\mathbb{Z}}}^{-1}}$ will be called
the \emph{stationary coding} of measure $\mu$ w.r.t.\ function $f$.

Assuming that (\ref{UnZeta}) is satisfied, this proposition has been
proved in \cite{Debowski10}:
\begin{theorem}
  \label{theoConj}
  Let $\mu=P((X_i)_{i\in\mathbb{Z}}\in\cdot)$ be the distribution of
  the process (\ref{exUDPi}) where variables $K_i$ are IID and satisfy
  (\ref{ZetaK}). Put $\mathbb{Y}=\klam{0,1,2}$ and consider a~function
  $f:\mathbb{X}\mapsto \mathbb{Y}^+$ given as
  \begin{align}
    \label{ConjCode}
    f(k,z)=b(k)z2
    ,
  \end{align}
  where $1b(k)\in\klam{0,1}^+$ is the binary representation of
  a~natural number $k$. The process $(\bar Y_i)_{i\in\mathbb{Z}}$
  distributed according to the stationary coding of measure $\mu$
  w.r.t.\ function $f$, $\bar P((\bar
  Y_i)_{i\in\mathbb{Z}}\in\cdot)=\overline{\mu\circ
    \okra{f^{\mathbb{Z}}}^{-1}}$, satisfies conditions (a)--(d) for
  $\zeta(\beta^{-1})>4$. Variables $\bar Z_k$ may be constructed as
  $\bar Z_k=\bar s_k((\bar Y_i)_{i\in\mathbb{Z}})$ almost surely,
  where
  \begin{align}
    \bar s_k(w)
    :=
    \begin{cases}
      0 & \text{if $2b(k)02\subseq w$ 
        and $2b(k)12\not\subseq w$}, \\
      1 & \text{if $2b(k)12\subseq w$ 
        and $2b(k)02\not\subseq w$}, \\
      2 & \text{else}
    \end{cases} 
  \end{align}
  for $w\in\mathbb{Y}^{\mathbb{Z}}\cup\mathbb{Y}^*$.
\end{theorem} 

Inequality $\zeta(\beta^{-1})>4$ holds for $\beta>0.7728...$ and comes
from satisfying condition (c).  Processes $(\bar
Y_i)_{i\in\mathbb{Z}}$ and $(X_i)_{i\in\mathbb{Z}}$ live on different
probability spaces, say $(\Omega,\mathfrak{J},\bar P)$ and
$(\Omega,\mathfrak{J},P)$ respectively.

\appendix
\renewcommand{\thesection}{\Alph{section}}

\section{Excess-bounding lemma}
\label{ssecExcessBound}

This proposition is a~variation of a~more specific statement in
\cite[Theorem 2]{Debowski06}:
\begin{lemma}
\label{theoExcessBound}
Consider a~function $G:\mathbb{N}\rightarrow\mathbb{R}$ such that
$\lim_k G(k)/k=0$ and $G(n)\ge 0$ for all but finitely many $n$. For
any $A\in\mathbb{N}$ and infinitely many $n$, we have $AG(n)- G(An)\ge
0$.
\end{lemma}
\begin{proof}
  We have the identity
  \begin{align*}
    \sum_{k=0}^{m-1} 
    \frac{AG(A^k n)-G(A^{k+1} n)}{A^{k+1}}=
    G(n)-n\cdot \frac{G(A^m n)}{A^m n}
    .
  \end{align*}
  For $m$ tending to infinity, this implies
  \begin{align*}
    \sum_{k=0}^\infty \frac{AG(A^k n)-G(A^{k+1} n)}{A^{k+1}}
    =
    G(n)
    .
  \end{align*}
  Putting $n=A^p$, we obtain
  \begin{align}
    \label{ESeries}
    \sum_{k=p}^\infty \frac{AG(A^k)-G(A^{k+1})}{A^{k+1}}\ge 0
  \end{align}
  for all but finitely many $p$.  On the other hand, if
  $AG(n)- G(An)\ge 0$ did not hold for infinitely many $n$ then the
  sum in (\ref{ESeries}) would be strictly negative for all
  sufficiently large $p$.
 \end{proof}
 
 Two specific applications of that lemma, namely implications
 (\ref{DiffFFBeta}) and (\ref{DiffFGBeta}), are used in this paper.
 Firstly, consider functions $G_1(n)\ge G_2(n)\ge 0$ and
 $F_i(n)=2G_i(n)-G_i(2n)$. If the functions $G_i$ have equal limits
 $\lim_n G_i(n)/n=g<\infty$ then
\begin{align}
  \label{DiffFF}
  \limsup_{n\rightarrow\infty} & \kwad{F_2(n) - F_1(n)}\ge 0  
\end{align}
follows from Lemma \ref{theoExcessBound}.  From inequality
(\ref{DiffFF}) and inequality
$\limsup_{n}\, (a_n+b_n)\ge \limsup_{n} a_n+ \liminf_{n} b_n$,
we obtain
\begin{align}
  \label{DiffFFBeta}
  \liminf_{n\rightarrow\infty} \frac{F_1(n)}{n^\beta}>0
  \implies 
  \limsup_{n\rightarrow\infty} \frac{F_2(n)}{n^\beta}>0
  .
\end{align}
Secondly, if $G_1(n) = gn + \tilde G_1(n)$ then
\begin{align}
  \label{DiffFGBeta}
  \liminf_{n\rightarrow\infty} \frac{\tilde G_1(n)}{n^\beta}>0
  \implies 
  \limsup_{n\rightarrow\infty} \frac{F_2(n)}{n^\beta}>0
  .
\end{align}
A~brief justification of the latter statement is as follows. The
left-hand side implies that $G(n)=G_2(n)-gn-Bn^\beta\ge 0$ for all but
finitely $n$ and a~certain $B>0$. Then it suffices to apply Lemma
\ref{theoExcessBound} to obtain the right-hand side.

\section{A bound for the longest repeat}
\label{ssecMaximalRepeat}

Let us prove a bound for the maximal length of a~repeat, defined
in (\ref{DefL}):
\begin{lemma}
  \label{theoFEL}
  For a~finite-energy process $(X_i)_{i\in\mathbb{Z}}$ over
  a~countable alphabet $\mathbb{X}$,
  \begin{align}
   \label{FELExp}
   \sup_{n\ge 2}
   \sred{\okra{\frac{\mathbf{L}(X_{1:n})}{\log n}}^{q}}
   &<\infty
   ,
   \quad 
   q>0
   .
 \end{align}
\end{lemma}
\emph{Remark:} The almost sure version of bound (\ref{FELExp}) was
shown for finite-energy processes over a~finite alphabet by Shields
\cite{Shields97}. Such bound does not hold for stationary processes in
general, cf. \cite{Shields92b}.

\begin{proof}
  Assume (\ref{FE}) and consider a~$j>i\ge 0$. Applying the idea from
  \cite{KontoyiannisSuhov94}, let us notice that given $X_{1:j}=w$
  condition $X_{j+1:j+k}=X_{i+1:i+k}$ is equivalent to $X_{j+1:j+k}=u$
  for a~certain string $u$ that is a~function of $w$. Thus
  \begin{align*}
    &P(X_{j+1:j+k}=X_{i+1:i+k})
    \\
    &=\sum_{w\in\mathbb{X}^j}
    P(X_{j+1:j+k}=X_{i+1:i+k}|X_{1:j}=w) P(X_{1:j}=w)
    \\
    &\le \sum_{w\in\mathbb{X}^j}Kc^k P(X_{1:j}=w) =Kc^k
    .
  \end{align*}
  Hence
  \begin{align*}
    P(\mathbf{L}(X_{1:n})\ge k)
    &= P(\exists_{0\le i<j\le n-k} X_{j+1:j+k}=X_{i+1:i+k})
    \\
    &\le 
    \sum_{0\le i<j\le n-k} P(X_{j+1:j+k}=X_{i+1:i+k})
    \\
    &\le
    \frac{(n-k)(n-k-1)}{2} Kc^k
    \le \frac{n^2 Kc^k}{2}
    .
  \end{align*}
  This bound is nontrivial for $k>A:=(2\log n+\log K-\log 2)/\log
  c^{-1}$. Consider a~sufficiently large $n$ so that $A\ge 1$. Then
  inequality (\ref{FELExp}) follows from the series of inequalities
  \begin{align*}
    \sred{(\mathbf{L}(X_{1:n}))^q}
    &\le A^{q} +\sum_{k> A}  k^q P(\mathbf{L}(X_{1:n})\ge k)
    \\
    &\le \sum_{k=0}^\infty (k+A)^q c^{k} \le  A^q \sum_{k=0}^\infty (k+1)^q c^{k}
    ,
  \end{align*}
  where 
  $\sum_{k=0}^\infty (k+1)^q c^{k}< \infty$.
\end{proof}


\section*{Acknowledgment}

The author wishes to thank Peter Harremo\"{e}s, Peter Gr\"unwald,
Alfonso Martinez, Jan Mielniczuk, Jacek Koronacki, Laurence Cantrill,
Paul Vitanyi, En-hui Yang, and an anonymous referee for discussing the
challenging composition of this paper.  Special thanks are due to
James Crutchfield, Arthur Ramer, and Peter Gr\"unwald for inviting the
author's to the Santa Fe Institute, the University of New South Wales,
and the Centrum Wiskunde \& Informatica, where important constructions
for this paper were completed. Finally, we thank cordially Prof.\
Emerit.\ Gabriel Altmann for many words of support during our
investigations.




 


\end{document}